\documentclass[12pt]{amsart}
\usepackage[margin=1in]{geometry}
\usepackage{amssymb,amsmath,amsthm,bbold,mathtools,bm} 
\def\-{\raisebox{.75pt}{-}}
\usepackage{enumitem}
\usepackage{comment}
\usepackage{tikz} 
\usetikzlibrary{cd}
\usetikzlibrary{arrows}

\numberwithin{table}{section}
\numberwithin{equation}{section}

\newtheorem{rmk}{Remark}[section]

\newtheorem{thm}{Theorem}[section]
\newtheorem*{thm*}{Theorem}
\newtheorem{cor}{Corollary}[section]
\newtheorem{obs}{Observation}[section]
\theoremstyle{plain}

\newcommand{\R}{\mathbb{R}}
\newcommand{\C}{\mathbb{C}}

\newcommand{\U}{\mathrm{U}}
\newcommand{\ut}{\mathfrak{u}}
\newcommand{\dg}{\mathrm{d}_{\mathfrak{g}}}
\newcommand{\du}{\mathrm{d}_{\mathfrak{u}(1)}}
\newcommand{\Ad}{\mathrm{Ad}}
\newcommand{\diag}{\mathrm{diag}}
\newcommand{\crit}{\mathrm{crit}\,}
\newcommand{\sM}{\mathcal{M}}

\newcommand{\sL}{\mathcal{L}}
\newcommand{\sF}{\mathcal{F}}

\newcommand{\pf}{\mathrm{pf}\,}
\newcommand{\Hess}{\mathrm{Hess}\,}
\newcommand{\sgn}{\mathrm{sgn}\,}

\newcommand{\vol}{\mathrm{vol}}
\newcommand{\dive}{\mathrm{div}}
\newcommand{\Sym}{\mathrm{Sym}}
\newcommand{\MV}{\mathcal{V}}

\usepackage{hyperref,cleveref}
\usepackage[alphabetic,initials]{amsrefs}

\title{Equivariant localization in Batalin-Vilkovisky formalism}
\author{Alberto S. Cattaneo}
\address{
	Department of Mathematics, University of Zurich, Winterthurerstrasse 190, CH-8057 Zurich, Switzerland
	}
\email{cattaneo@math.uzh.ch}
\author{Shuhan Jiang}
\address{
	Department of Mathematics, University of Zurich, Winterthurerstrasse 190, CH-8057 Zurich, Switzerland
	}
\email{shuhan.jiang@math.uzh.ch}
\thanks{ASC acknowledges partial support of the SNF Grant No. 200021 227719 and of the Simons Collaboration on Global Categorical Symmetries. This research was (partly) supported by the NCCR SwissMAP, funded by the Swiss National Science Foundation. This material is based upon work from COST Action 21109 CaLISTA, supported by COST (European Cooperation in Science and Technology)(www.cost.eu). This material is based also upon work supported by the Swedish Research Council under grant no. 2021-06594 while the authors ASC and SJ were in residence at Institut Mittag-Leffler in Djursholm, Sweden during the year of 2025.}

\begin{document}
	
	\begin{abstract}
		We derive equivariant localization formulas of Atiyah--Bott and cohomological field theory types in the Batalin-Vilkovisky formalism and discuss their applications in Poisson geometry and quantum field theory.
	\end{abstract}
	
	\maketitle
	
\section{Introduction}

Let $G$ be a compact Lie group with Lie algebra $\mathfrak{g}$.
Let $M$ be a manifold endowed with a left $G$-action. 
Let $\Omega_G(M)=(\Omega(M) \otimes \Sym(\mathfrak{g}^{\vee}))^G$ denote the graded commutative algebra of $G$-equivariant differential forms on $M$, i.e., differential forms $\alpha$ on $M$ with values in the ring of polynomials over $\mathfrak{g}$, such that
\[
\alpha(\Ad_g \xi) = L_{g^{-1}}^* \alpha(\xi), \quad \forall g \in G,~\forall \xi \in \mathfrak{g},
\] 
where $\Ad$ is the adjoint action of $G$ on $\mathfrak{g}$, $L_g\colon M \rightarrow M$ is the left multiplication by $g \in G$. The grading on $\Omega_G(M)$ is given by 
\[
\Omega_G^p(M) = \bigoplus_{r+2q=p} (\Omega^r(M) \otimes \Sym^q(\mathfrak{g}^{\vee}))^G.
\]

The equivariant differential $\dg\colon \Omega_G^p(M) \rightarrow \Omega_G^{p+1}(M)$ is defined as
\[
\dg \alpha(\xi) = (\mathrm{d} - \iota_{X_{\xi}})\alpha(\xi),
\]
where $X_{\xi}(p) = \frac{d}{dt} (p \cdot \exp(t\xi))|_{t=0}$ is the fundamental vector field generated by $\xi \in \mathfrak{g}$. Fixing a basis $\{\xi_a\}$ of $\mathfrak{g}$, the differential $\dg$ can be written as
\[
\dg = \mathrm{d} - \phi^a \iota_{X_a},
\]
where $X_a = X_{\xi_a}$ and $\{\phi^a\}$ is the dual basis of $\mathfrak{g}^{\vee}$.

If $M$ is compact and oriented, we can intergrate equivariant differential forms on $M$ to obtain a map:
\[
\int_M\colon \Omega_G(M) \rightarrow \Sym(\mathfrak{g}^{\vee})^G.
\]
Stokes' theorem can be easily generalized to this equivariant setting: Let $\alpha \in \Omega_G(M)$. If $\dg\alpha = 0$, then the integral
$
\int_M \alpha
$
depends only on $[\alpha] \in H_{G}(M)$, where $H_G(M)$ is the cohomology group of $(\Omega_G(M), \dg)$.

The Equivariant Localization Principle is the following simple observation:
\begin{obs}\label{oelp}
	Let $\gamma$ be a $G$-equivariant $1$-form on $M$.  Let $\alpha \in \Omega_G(M)$. If $\dg\alpha = 0$, then the integral
	\begin{equation}\label{elp}
		Z_{\gamma}[t](\alpha)\coloneqq\int_M \alpha e^{\mathrm{i}t \dg\gamma} = \int_M \alpha e^{\mathrm{i}t\left(\mathrm{d} \gamma -\phi^a \gamma(X_a)\right)}
	\end{equation}
	is independent of $t \in \R$. In particular, $Z_{\gamma}[t](\alpha) = Z_{\gamma}[0](\alpha) = \int_M \alpha$.
\end{obs}
\begin{rmk}
	Note that $\alpha e^{\mathrm{i}t \dg\gamma} = e^{\mathrm{i}t[\dg, \gamma \wedge]} \alpha$.
\end{rmk}
The basic idea of equivariant localization is, instead of computing the integral $\int_M \alpha$ directly, one should compute the superficially more complicated integral $Z_{\gamma}[t](\alpha)$ in the limit $t \to \infty$, which will ``localize'' to a integral over the zero locus of $\gamma(\eta)(X_{\eta})$, $\eta \in \mathfrak{g}$. 
Different choices of $\gamma$ then lead to different localization formulas. Let us consider the following two strategies for selecting $\gamma$. 
\begin{rmk}
	For simplicity, we will take $G$ to be $\mathrm{U}(1)$. $\Omega_{\U(1)}(M)$ consists of $U(1)$-invariant differential forms on $M$ with values in $\R[\phi]$, where $\phi \in \ut(1)^{\vee}$ is the dual of $\mathrm{i} \in \ut(1) \cong \mathrm{i} \R$. The corresponding equivariant differential can be expressed as
	\[
	\du = \mathrm{d} - \phi \iota_X,
	\]
	where $X(p)\coloneqq\frac{d}{dt} (p \cdot \exp(\mathrm{i}t))|_{t=0}$. 
	%The discussion below can be effortlessly generalized for a general compact Lie group $G$.
\end{rmk}
\begin{itemize}
	\item Equipping $M$ with a $\U(1)$-invariant Riemannian metric $g$, we choose $\gamma$ to be
	\begin{equation}\label{gm1}
		\gamma = X^{\flat}\coloneqq g(X, \cdot).
	\end{equation}
	It follows that
	\[
	\du \gamma =  g(\nabla X, \cdot) + g(X, \nabla(\cdot)) - \phi g(X,X),
	\]
	where $\nabla$ is the Levi-Civita connection of $g$. The zero locus $M_X$ of $\gamma(X)=g(X,X)$ consists of the fixed points of the $\U(1)$-action. 
	%$M_X$ is a submanifold of $M$, whose connected components may have different dimensions.
	Computing $Z_{\gamma}[t](\alpha)$ in the limit $t \to \infty$ yields the Berline--Vergne localization theorem \cite{berline1982classes} when $M_X$ is discrete. More generally, it leads to the Atiyah--Bott localization theorem \cite{atiyah1984moment} when the connected components of $M_X$ have positive dimensions. For more details, we refer the reader to \cite{alekseev2000notes} and \cite[Chapter~7]{berline2003heat}.
	
	\item Let $W$ be a vector space endowed with a left linear action $\rho$ of $\U(1)$ and a $\U(1)$-invariant inner product $h(\cdot,\cdot)$. Let $\mathcal{F}\colon M \rightarrow W$ be a $\U(1)$-equivariant smooth map. We choose $\gamma$ to be
	\begin{equation}\label{gm2}
		\gamma =  h(\dot{\rho}(\mathrm{i}) \mathcal{F}, \mathrm{d} \mathcal{F}),
	\end{equation}
	where $\dot{\rho}$ is the induced $\mathfrak{u}(1)$-action on $W$. It follows that
	\[
	\du \gamma = h(\dot{\rho}(\mathrm{i}) \mathrm{d} \mathcal{F}, \mathrm{d} \mathcal{F}) - \phi h(\dot{\rho}(\mathrm{i})\mathcal{F}, \dot{\rho}(\mathrm{i})\mathcal{F}),
	\]
	because 
	\begin{align}\label{equivf}
		\iota_X \mathrm{d} \mathcal{F} = \mathrm{Lie}_X \mathcal{F} = \dot{\rho}(\mathrm{i}) \mathcal{F}.
	\end{align}
	If $\U(1)$ acts freely on $W-\{0\}$, then the zero locus $M_{\mathcal{F}}$ of $\gamma(X)=h(\dot{\rho}(\mathrm{i})\mathcal{F}, \dot{\rho}(\mathrm{i})\mathcal{F})$ is exactly the zero locus of $\mathcal{F}$. \eqref{equivf} also implies that
	$
	M_{X} \subset M_{\sF}.
	$
	If $M_{\sF}$ is discrete, we also have
	$
	M_{\sF} \subset M_{X},
	$
	since $g M_{\sF} = M_{\sF}$ for all $g \in \U(1)$.
	
	In particular, let us consider $W=\C^k$, equipped with the standard inner product and the following $\U(1)$-action:
	\[
	e^{\mathrm{i}\theta}(f_1, \dots, f_k) = (e^{\mathrm{i}w_1\theta}f_1, \dots e^{\mathrm{i}w_k\theta}f_k), \quad e^{\mathrm{i}\theta} \in \U(1),~ (f_1, \dots, f_k) \in \C^k,
	\]
	where $w_1, \dots, w_k$ are non-zero integers. Denote $\sF=(\sF^1, \cdots, \sF^k)$, where $\sF^1, \dots, \sF^k$ are complex functions on $M$. If $\sF$ has isolated zeros, then $M_{\sF} = M_{X}$ and $\dim M$ is even. Computing $Z_{\gamma}[t](\alpha)$ in the limit $t \to \infty$ leads to the following theorem:
	\begin{thm}
		Let $\alpha \in \Omega_{\U(1)}(M)$.  If $d_{\ut(1)} \alpha = 0$, then
		\begin{equation}\label{cohft1}
			\int_M \alpha = (-2\pi)^{\dim(M)/2} \sum_{p \in M_{X}} \frac{\pf \mathrm{Im} (\sum_{l=1}^k w_l \partial_i \overline{\sF^l} \partial_j \sF^l)(p)}{\sqrt{\det \mathrm{Re} (\sum_{l=1}^k w_l^2 \partial_i \overline{\sF^l} \partial_j \sF^l)(p)}} \alpha_{[0]}(p),
		\end{equation}
		where $\partial_i \sF$ are the partial derivatives $\sF$ computed in some local coordinates of $M$, $\alpha_{[0]}$ is the $0$-form component of $\alpha$.
		%\eqref{cohft1} does not depend on the local coordinates since $M$ is orientable.
		%, and the determinant $\det$ and pfaffian $\pf$.
	\end{thm}
\end{itemize}
\begin{rmk}
	We refer to the localization induced by \eqref{gm1} as the Atiyah--Bott type localization and refer to the localization induced by \eqref{gm2} as the cohomological field theory (CohFT) type localization. The terminology is due to the fact that, for CohFT type localization, $Z_{\gamma}[t]$ can be viewed as the partition function of a zero dimensional CohFT. 
\end{rmk}

In this work, we prove the Berline--Vergne localization theorem and \eqref{cohft1} using the standard stationary phase approximation for Morse functions, applied in the setting of equivariant multivector fields rather than equivariant differential forms. We then extend the equivariant localization principle \eqref{oelp} to the Batalin-Vilkovisky formalism and apply the stationary phase approximation for Morse-Bott functions to establish a BV analog of the Atiyah--Bott localization theorem and to generalize \eqref{cohft1} (cf. Theorems \ref{pvloc} and \ref{pvloccoh}). In the final section, we discuss applications of our localization theorems to Poisson geometry and quantum field theory. Specifically, we derive a Poisson-geometric counterpart of the Duistermaat-Heckman formula \cite{duistermaat1982variation} for unimodular Poisson manifolds.

\section{Equivariant localization principle in BV formalism}

%\subsection{The BV algebra of multivector fields}

Let $M$ be an orientable manifold. Let $\MV(M)=\Gamma(\Lambda TM)$ denote the space of multivector fields over $M$. $\MV(M)$ is a graded commutative algebra over $\R$. The grading of $\MV(M)$ is given by  
\[
\MV^{p}(M)\coloneqq \Gamma(\Lambda^{-p} TM)
\] 
for $-\dim M \leq p \leq 0$ and $\MV^{p}(M)=0$ for all other values of $p$. The graded commutative product on $\MV(M)$ is the wedge product $\wedge$ between multivector fields.
There exists a natural graded Lie superbracket of degree $1$ on $\MV(M)$, called the Schouten--Nijenhuis bracket, which is determined by the following properties:
\begin{enumerate}
	\item $\{f,g\}=0$ for all $f, g \in \MV^0(M)=C^{\infty}(M)$;
	\item $\{X, \cdot\}=\mathrm{Lie}_X$ for all $X \in \MV^{-1}(M)=\mathfrak{X}(M)$;
	%where $\mathrm{Lie}_X P$ is the Lie derivative of $P \in \MV(M)$ along $X$
	%\item $\{P_1, P_2\}=-(-1)^{(|P_1|-1)(|P_2|-1)}\{P_2,P_1\}$ for all $P_1, P_2 \in \MV(M)$;
	\item $\{P_1, P_2 \wedge P_3\} = \{P_1,P_2\} \wedge P_3 + (-1)^{|P_2|(|P_1|-1)}P_2 \wedge \{P_1,P_3\}$ for all $P_1, P_2, P_3 \in \MV(M)$.
\end{enumerate}
%The last two properties above guarantees that 
$\MV(M)$ equipped with the Schouten-Nijenhuis bracket is a Gerstenhaber algebra. 

Let $\Omega(M)=\Gamma(\Lambda T^*M)$ denote the space of differential forms on $M$. There exists a canonical paring $\langle \cdot, \cdot \rangle\colon \Omega(M) \times \MV(M) \rightarrow C^{\infty}(M)$, given by
\[
\langle \alpha, P \rangle = \alpha(x)(P(x)), \quad \alpha \in \Omega(M),~P \in \MV(M).
\]
The right contraction of a $p$-differential form $\alpha$ by a $q$-multivector field $P$ is a $(p-q)$-differential form $\alpha \lrcorner P$ is defined by the relations
\[
\langle \alpha \lrcorner P, Q \rangle = \langle \alpha, P \wedge Q \rangle, \quad \forall Q \in \MV^{p-q}(M).
\]
For $X \in \MV^{-1}(M)$, $\lrcorner X$ is the usual right contraction by the vector field $X$. In general,  $\lrcorner P$ is not a superderivation, and we have
$
\lrcorner (P \wedge Q) = \lrcorner Q \lrcorner P.
$
The left contraction by $P$, denoted as $P \llcorner $ (or $\iota_P$), is defined in a similar manner. Likewise, we also have
$
(P \wedge Q) \llcorner = P \llcorner  Q \llcorner.
$
%\[
%P \llcorner \alpha = (-1)^{|P||\alpha|} \alpha \lrcorner P.
%\]

Let $\vol \in \Omega^{top}(M)$ be a volume form on $M$. Let $\Delta\colon \MV(M) \rightarrow \MV(M)$ be the degree $1$ $\R$-linear map defined by 
\begin{equation*}
	\begin{tikzcd}
		\MV(M) \arrow[r, "\llcorner \vol"] \arrow[d, "\Delta"] & \Omega(M) \arrow[d, "\mathrm{d}"] \\
		\MV(M) \arrow[r, "\llcorner\vol"]  & \Omega(M)
	\end{tikzcd},
\end{equation*}
where $\mathrm{d} \colon \Omega(M) \rightarrow \Omega(M)$ is the de Rham differential of $M$. By definition, $\Delta^2 = 0$. One can also show that $\Delta$ generates the Schouten--Nijenhuis bracket. That is,
\[
\Delta(P \wedge Q) = \Delta(P) \wedge Q + (-1)^{|P|} P \wedge \Delta(Q) + (-1)^{|P|}\{P,Q\}.
\]
For $X \in \MV^{-1}(M)=\mathfrak{X}(M)$, we have
\begin{align*}
	\Delta(X) = \vol^{-1} \Delta(X) \llcorner \vol = \vol^{-1} \mathrm{d} (X \llcorner \vol)) = \dive_{\vol}(X).
\end{align*}
Thus, $\Delta$ generalizes the usual divergence operator on vector fields. $\MV(M)$ equipped with $\{\cdot, \cdot\}$ and $\Delta$ is a BV algebra (see \cite{cattaneo2018graded} and references therein).

From a graded geometric point of view, $\MV(M)$ can be identified with the algebra of superfunctions on the graded manifold $T^*[-1] M$, which admits a canonical odd symplectic structure $\omega_{st}$ of degree $-1$. In local coordinates $(x^i, \xi_i=\partial_i)$, 
$
\omega_{st} = dx^i \wedge d\xi_i.
$
The odd Poisson bracket $\{\cdot, \cdot\}$ associated to $\omega$ is exactly the Schouten-Nijenhuis bracket. 

Let $(\sM, \omega)$ be a general odd symplectic manifold of degree $-1$. A Berezinian $\mu$ is said to be compatible with $\omega$ if there exists an atlas of Darboux charts of $\mathcal{M}$ such that locally,
$
\mu = d^nx d^n \xi.
$
(We allow the coordinate functions $x^i$ to have nonzero degree $d(x^i)$. The corresponding degree of the anti-coordinates $\xi_i$ are $-1-d(x^i)$.) The BV Laplacian $\Delta$ is defined locally by the following formula \cite{schwarz1993geometry}:
\[
\Delta = \frac{\partial}{\partial x^i}\frac{\partial}{\partial \xi_i}.
\]
Alternatively, $\Delta$ can be defined as
\[
\Delta(f) \coloneqq \frac{1}{2} \dive_{\mu}(X_f), \quad f \in C^{\infty}(M),
\]
where $X_f$ is the Hamiltonian vector field of $f$, defined via
$
\iota_{X_f} \omega = \mathrm{d} f,
$
or equivalently, $X_f=(-1)^{\deg f}\{f, \cdot\}$, $\dive_{\mu}(X)$ is the divergence of a vector field $X$ over $M$, defined via
$
\int_{\sM} \mu X(f) = -\int_{\sM} \mu \dive_{\mu}(X)f.
$
$C^{\infty}(\sM)$ endowed with $\{\cdot,\cdot\}$ and $\Delta$ is a BV algebra.

\vspace{10pt}
Let us reformulate the equivariant localization principle within the BV formalism. This will be done in two steps: first for $(T^*[-1]M, \omega_{st})$, and then for a general odd symplectic manifold $(\mathcal{M}, \omega)$ of degree $-1$.

Suppose that we have an action of a Lie group $G$ on $M$. Let
\[
\MV_{G}(M) = (\MV(M) \otimes \Sym (\mathfrak{g}^{\vee}))^{G}
\]
denote the graded commutative algebra of equivariant multivector fields over $M$. We define the degree of a $q$-homogeneous polynomial valued $r$-multivector field to be $2q - r$. Let $\vol$ be a $G$-invariant volume form on $M$. Under the identification 
$
\MV(M) \overset{\llcorner \vol}{\cong} \Omega(M),
$
the left contraction $X \llcorner$ on $\Omega(M)$ becomes a left multiplication 
$
X \wedge \cdot
$
on $\MV(M)$.
We define the equivariant differential $\Delta_{\mathfrak{g}}\colon \MV^p_{G}(M) \rightarrow \MV^{p+1}_{G}(M)$ as
\[
\Delta_{\mathfrak{g}}\coloneqq \Delta - \phi^a X_a \wedge \cdot.
\]
Note that neither $\Delta$ nor $X_a \wedge \cdot$ are derivations of $\MV(M)$. However, their commutator
\[
[\Delta, X_a \wedge \cdot](P) = \Delta(X_a \wedge P) + X_a \wedge \Delta(P)= \dive_{\vol}(X_a) \wedge P - \{X_a, P\} = - \mathrm{Lie}_X(P)
\] 
is a derivation since $\dive_{\vol}(X_a)=0$ by our choice of $\vol$. 
(Another way to see this is to observe that $\dg(P \llcorner \vol) = (\Delta_{\mathfrak{g}} P) \llcorner \vol$.)
We have the desired formula:
$
\Delta_{\mathfrak{g}}^2 = \phi^a \mathrm{Lie}_{X_a}.
$

If $M$ is compact, we can integrate equivariant multivector fields over $M$ to obtain a map
\[
\int_M\colon \MV_G(M) \rightarrow \Sym(\mathfrak{g}^{\vee})^G,
\]
where $\int_M P$ is understood as
\[
\int_M P\coloneqq\int \vol \lrcorner P = \int_M \vol P_{[0]},
\]
where the function $P_{[0]}$ is the $0$-multivector field component of $P$. The divergence theorem can be easily generalized to this equivariant setting: Let $P \in \MV_G(M)$. If $\Delta_{\mathfrak{g}} P = 0$, then the integral
$
\int_M P
$
depends only on $[P] \in H_{G}(M)$, where $H_G(M)$ is the cohomology group of $(\MV_G(M), \Delta_{\mathfrak{g}})$.

The Equivariant Localization Principle \ref{oelp} can be reformulated as follows:
\begin{obs}\label{obvelp}
	Let $\gamma$ be a $G$-equivariant $1$-form on $M$.  Let $P \in \MV_G(M)$. If $\Delta_{\mathfrak{g}} P = 0$, then the integral
	\begin{equation}\label{bvelp}
		Z_{\gamma}[t](P)\coloneqq\int_M e^{\mathrm{i}t[\Delta_{\mathfrak{g}}, \gamma \llcorner]}(P) = \int_M e^{-\mathrm{i}t(\phi^a \gamma(X_a))}e^{\mathrm{i}t[\Delta, \gamma \llcorner]}(P),
	\end{equation}
	is independent of $t \in \R$, where $\gamma \llcorner$ is the left contraction by $\gamma$. In particular, $Z_{\gamma}[t](P) = \int_M P$.
\end{obs}

Let us now consider a general odd symplectic manifold $(\sM, \omega)$ of degree $-1$ endowed with a left Hamiltonian $G$-action. (Here $G$ is just an ordinary Lie group.) Let $\mu$ be a Berezinian on $\sM$ that is compatible with $\omega$ and the $G$-action.  Let
\[
C^{\infty}_G(\sM) = (C^{\infty}(\sM) \otimes \Sym (\mathfrak{g}^{\vee}))^{G}.
\]
As usual, we assign degree $2$ to an element of $\Sym^1 (\mathfrak{g}^{\vee})$. The degree of an element of $C^{\infty}_G(\mathcal{M})$ is then defined as the corresponding total degree.
We define the equivariant differential $\Delta_{\mathfrak{g}}\colon C^{\infty}(\sM)^p_{G}(M) \rightarrow C^{\infty}(\sM)^{p+1}_{G}(M)$ as
\[
\Delta_{\mathfrak{g}}\coloneqq \Delta - \phi^a f_{X_a} \cdot,
\]
where $f_{X_a}$ is the Hamiltonian of the fundamental vector fields $X_a$. (Such $f_{X_a}$ is unique because it is of odd degree.) One can easily check that
$
\Delta_{\mathfrak{g}}^2 = \phi^a X_a.
$

If the body $\sM_{red}$ of $\sM$ is compact, we can perform the following BV integral:
\[
\int_{\sL \subset \sM} \sqrt{\mu|_{\sL}} f|_{\sL}, \quad f \in C^{\infty}(\sM),
\]
where $\sL$ is a Lagrangian sub-supermanifold of $\sM$, $\sqrt{\mu|_{\sL}}$ is the density on $\sL$ induced by $\mu$. The BV integral can be extended to an integral of equivariant super functions on $\sM$:
\[
\int_{\sL \subset \sM}  \sqrt{\mu|_{\sL}} \colon C^{\infty}_G(\sM)\rightarrow \Sym(\mathfrak{g}^{\vee})^G,
\]
The BV Stokes' theorem can also be easily generalized to the equivariant BV setting: Let $f \in C^{\infty}_G(\sM)$. Let $\sL_s$ be a smooth family of Lagrangian sub-supermanifold of $\sM$, $s \in [0,1]$. Assume that $\sL_s$ is $G$-invariant, i.e., the restriction $f_{X_{\xi}}|_{\sL_s}$ vanishes for all $\xi \in \mathfrak{g}$. If $\Delta_{\mathfrak{g}} f =0$, then the equivariant BV integral
\[
\int_{\sL_s \subset \sM} \sqrt{\mu|_{\sL_s}} f|_{\sL_s}
\]
depends only on $[f] \in H_{\Delta_{\mathfrak{g}}}(\mathcal{M})$ and does not depend on $s$.

The Equivariant Localization Principle \ref{obvelp} can be generalized as follows:
\begin{obs}\label{obvelp2}
	Let $Y \in \mathfrak{X}(\sM) \otimes \Sym(\mathfrak{g}^{\vee})$ be a $G$-equivariant odd vector field over $\sM$, i.e., $[X_{\xi}, Y] = [\xi, Y]$ for all $\xi \in \mathfrak{g}$. We have
	\[
	[\Delta_{\mathfrak{g}}, Y] = [\Delta, Y] - \phi^a Y(f_{X_a}) = [\Delta, Y] - \phi^a \omega(X_a, Y),
	\]
	where we use $Y(f_{X_a}) = \iota_Y df_{X_a} = \iota_Y \iota_{X_a} \omega = \omega(X_a, Y)$. Assume that
	\begin{enumerate}
		\item $[\Delta,Y]$ is nilpotent;
		\item $[\Delta,Y]$ commutes with $\phi^a \omega(X_a, Y)$.
	\end{enumerate}
	Let $f \in C^{\infty}_G(\sM)$. If $\Delta_{\mathfrak{g}} f = 0$, then the integral
	\begin{equation}\label{bvelp2}
		Z_{Y, \sL}[t](f):= \int_{\mathcal{L} \subset \sM} \sqrt{\mu|_{\sL}} \left(e^{\mathrm{i}t[\Delta_{\mathfrak{g}}, Y]}(f)\right)|_{\sL} =  \int_{\mathcal{L} \subset \sM} \sqrt{\mu|_{\sL}} \left(e^{-\mathrm{i}t(\phi^a \omega(Y, X_a))}e^{\mathrm{i}t[\Delta,Y]}(f)\right)|_{\sL},
	\end{equation}
	is independent of $t \in \R$, where $\sL$ be a $G$-invariant Lagrangian sub-supermanifold of $\sM$. In particular, $Z_{Y,\sL}[t](f) = \int_{\mathcal{L} \subset \sM} \sqrt{\mu|_{\sL}}  f_{\sL}$.
\end{obs}

Combined with the stationary phase method for supermanifolds \cite[Theorem 4.2.2]{zakharevich2017localization}, Principle \ref{obvelp2} can be used to derive localization formulas for equivariant BV integrals; however, this will not be discussed in the present work.

\section{Equivariant localization theorems}

For simplicity, let us consider $G = \U(1)$. 

Let $M$ be an orientable compact manifold of dimension $n$. Let us equip $M$ with a $\mathrm{U}(1)$-invariant Riemannian metric $g$ and choose $\vol$ to be the volume form of $g$. Let $S$ be a Morse--Bott function on $M$, i.e., the critical locus $\crit S$ of $S$ is a closed submanifold of $M$ and the Hessian $\Hess S$ of $S$ is non-degenerate along the normal bundle $N_{\crit S}$ of $\crit S$. Recall the stationary phase approximation:
\[
\int_M \vol f e^{\mathrm{i}tS} \sim_{t \to \infty} \left(\frac{2\pi}{t}\right)^{k/2} e^{\mathrm{i}\frac{\pi}{4} \sgn \Hess S|_{N_{\crit S}}}\left( \int_{\crit S}  f e^{\mathrm{i}tS}  \frac{\vol}{|\det \Hess S |_{N_{\crit S}}|^{1/2}} + O(t^{-1})\right),
\]
where $k$ is the codimension of $\crit S$, $\sgn \Hess S|_{N_{\crit S}}$ is the signature of the Hessian matrix, and $\vol / |\det \Hess S |_{N_{\crit S}}|^{1/2}$ defines a density over $\crit S$.

\subsection{Atiyah--Bott type localization}

Let $\gamma=X^{\flat}\coloneqq g(X, \cdot)$ be the metric dual of $X$. Taking the limit $t \to \infty$ localizes the integral 
\[
Z_{\gamma}[t](P) = \int_M e^{-\mathrm{i}tg(X,X)} e^{\mathrm{i}t[\Delta ,X^{\flat} \llcorner]}(P)
\]
to the fixed point set $M_{X}$ of the $\U(1)$-action.  %Let $S=\gamma(X)=g(X,X)$.

\begin{thm}\label{pvloc}
	Let $P \in \MV_{\U(1)}(M)$. If $\Delta_{\ut(1)} P = 0$ and $M_X$ is of codimension $2m$, then
	\begin{equation}\label{pvab}
		\int_M P = \frac{(-2\pi)^{m}}{m!} \int_{M_X} \langle (\nabla X^{\flat})^m, P_{[2m]} \rangle \frac{\vol}{\sqrt{\det \Hess g(X,X)|_{N_X}}},
	\end{equation}
	where $\nabla$ is the Levi-Civita connection of $g$ and $N_X$ is the normal bundle of  $M_X$. In particular, if $n=2m$, i.e., if $M_X$ is discrete, then
	\begin{equation}\label{bvf}
		\int_M P = (-2\pi)^m \sum_{p \in M_{X}}\frac{\langle \vol, P_{[2m]}\rangle (p)}{\lambda_1(p)\cdots\lambda_m(p)},
	\end{equation}
	where $\lambda_1(p), \dots, \lambda_m(p)$ are the weights of the induced $\U(1)$-action on $T_pM$.
\end{thm}
\begin{rmk}
	\eqref{bvf} is essentially the Berline--Vergne localization formula.%, as $\vol \lrcorner P$ is an equivariantly closed form on $M$.
\end{rmk}

\begin{proof}
	Since $[\Delta_{\mathfrak{u}(1)}, \iota_{X}]$ is invertible outside of $M_X$, the integrand $e^{-\mathrm{i}tg(X,X)} e^{\mathrm{i}t[\Delta ,X^{\flat} \llcorner]}(P)$ is $\Delta_{\mathfrak{u}(1)}$-exact outside of $M_X$. Applying the stationary phase approximation, we obtain
	\begin{align*}
		Z_{\gamma}[t](P) \sim_{t \to \infty} \left(\frac{2\pi}{t}\right)^m e^{\mathrm{i}m\pi/2} \int_{M_X}  f \frac{\vol}{\sqrt{\det \Hess g(X,X)|_{N_X}}} + O(t^{-1}),
	\end{align*}
	where 
	$
	f =  \frac{(\mathrm{i}t)^m}{m!} [\Delta, X^{\flat}\llcorner]^m(P_{[2m]}).
	$
	Since $Z_{\gamma}[t](P)$ is independent of $t$, we have
	\[
	Z_{\gamma}[t](P)  = \frac{(-2\pi)^m}{m!} \int_{M_X} [\Delta, X^{\flat}\llcorner]^m(P_{[2m]}) \frac{\vol}{\sqrt{\det \Hess g(X,X)|_{N_X}}}.
	\]
	Let us work in normal coordinates $x^1, \ldots, x^n$ around a fixed point $p$ of the $\U(1)$-action. In such coordinates, the BV Laplacian $\Delta$ has the following expression:
	\[ 
	 \Delta = dx^{\mu} \llcorner \nabla_{\partial_{\mu}}.
	\]
	It follows that 
	\[
	[\Delta, X^{\flat}\llcorner] = dx^{\mu} \llcorner \nabla_{\partial_{\mu}}(X^{\flat} \llcorner) + X^{\flat} \llcorner dx^{\mu} \llcorner \nabla_{\partial_{\mu}} = \nabla_{\partial_{\mu}} X^{\flat}_{\nu} dx^{\mu} \llcorner dx^{\nu} \llcorner = \nabla X^{\flat} \llcorner,\footnote{$\nabla X^{\flat}$ is a $2$-form since $X$ is a Killing vector field.}
	\]
	thus proving the first part of the theorem.
	If $M_X$ is discrete, then $X$ can be linearized around $p \in M_X$:
	\[
	X = \sum_{i=1}^{m} \lambda_i(p) (x^{2i} \partial_{2i-1} - x^{2i-1} \partial_{2i}),
	\]
	where $\lambda_1(p), \dots, \lambda_m(p)$ are non-zero integers. The Hessian of $g(X,X)$ at $p$ is the following diagonal matrix:
	\[
	\Hess g(X,X)(p) = 2\diag(\lambda_1^2, \lambda_1^2, \cdots, \lambda_m^2, \lambda_m^2).
	\]
	Therefore, $\sqrt{\det \Hess g(X,X)(p)} = 2^m \lambda_1^2 \cdots \lambda_m^2$. We also have
	\begin{align*}
		\nabla X^{\flat} (p) = 2 \sum_{i=1}^{m} \lambda_i(p) dx^{2i} \wedge dx^{2i-1}.
	\end{align*}
	It follows that
	\[
	(\nabla X^{\flat})^m (p) = 2^m m! \lambda_1(p) \cdots \lambda_m(p)dx^1 \wedge \cdots dx^{2m},
	\]
	thus proving the second part of the theorem.
\end{proof}

Recall the Atiyah--Bott localization theorem:
\[
	\int_M \alpha = \int_{M_X} \frac{\alpha}{e(N_X)},
\]
where $\alpha$ is an equivariantly closed form on $M$, and $e(N_X)$ is a representative of the equivariant Euler class of $N_X$. Letting $\alpha= \vol \lrcorner P$, we obtain
\begin{equation}\label{pvab2}
	\int_M P = \int_{M_X} \frac{\vol \lrcorner P}{e(N_X)}.
\end{equation}
\eqref{pvab2} is essentially different from our localization formula \eqref{pvab} when $\dim M_X > 0$. This is because the right-hand side of \eqref{pvab2} involves all $P_{[2k]}$ with $2k - \mathrm{codim} M_X \geq 0$, whereas the right-hand side of \eqref{pvab} only involves $P_{[\mathrm{codim} M_X]}$.

\subsection{CohFT type localization}

Let $\C^k$ be equipped with the standard inner product and a diagonal $\U(1)$-action $\rho$ with non-zero weights $w_1, \dots, w_k$. Let $\sF=(\sF^1, \dots, \sF^k)$ be a $\U(1)$-equivariant map from $M$ to $\C^k$. Let $\gamma= \mathrm{Re}(\overline{\dot{\rho}(\mathrm{i})\sF} d\sF)$. Taking the limit $t \to \infty$ localizes the integral 
\begin{align}\label{pi}
	Z_{\gamma}[t](P) = \int_M \vol e^{-\mathrm{i}t|\dot{\rho}(\mathrm{i})\sF|^2} e^{\mathrm{i}t[\Delta ,\iota_{\gamma}]}(P)
\end{align}
to the zero locus $M_{\sF}$ of $\sF$.
\begin{thm}\label{pvloccoh}
	Let $P \in \MV_{\U(1)}(M)$. If $\Delta_{\ut(1)} P = 0$ and $M_{\sF}$ has codimension $2m$, then
	\begin{equation}
		\int_M P = \frac{(-2\pi)^m}{m!} \int_{M_{\sF}} \langle (\sum_{l=1}^kw_l\mathrm{Im}(d\overline{\sF^l} \wedge \mathrm{d} \sF^l))^m, P_{[2m]} \rangle \frac{\vol}{\sqrt{\det \Hess(\sum_{l=1}^k w_l^2 |\sF^l|^2)|_{N_{\sF}}}}
	\end{equation}
   If $M_{\sF}$ is discrete, then $M_X  = M_{\sF}$ and $n=2m$, we have
   \begin{equation}
   	\int_M P = (-2\pi)^{m} \sum_{p \in M_{X}} \frac{\pf \mathrm{Im} (\sum_{l=1}^k w_l \partial_i \overline{\sF^l} \partial_j \sF^l)(p)}{\sqrt{\det \mathrm{Re} (\sum_{l=1}^k w_l^2 \partial_i \overline{\sF^l} \partial_j \sF^l)(p)}} \langle \vol, P_{[2m]} \rangle(p).
   \end{equation}
   %where $\langle \cdot, \cdot \rangle $ is the paring between $\Omega(M)$ and $\MV(M)$.
\end{thm}
\begin{proof}
	Applying the stationary phase approximation, we obtain
	\begin{align*}
		Z_{\gamma}[t](P)  = \frac{(-2\pi)^m}{m!} \int_{M_X} \langle (\nabla \gamma)^m, P_{[2m]} \rangle \frac{\vol}{\sqrt{\det \Hess |\dot{\rho}(\mathrm{i}) \sF|^2|_{N_X}}},
	\end{align*}
	where $\nabla$ is the Levi-Civita connection of $g$. A direct computation shows that
	\[
	\gamma = \sum_{l=1}^k w_l(\sF^l_1 \mathrm{d} \sF^l_2 - \sF^l_2 d\sF^l_1),
	\]
	where $\sF^l = \sF^l_1 + \mathrm{i} \sF^l_2$, $\sF^l_1$ and $\sF^l_2$ are real functions on $M$. It follows that
	\[
	\nabla \gamma(p) = 2\sum_{l=1}^k w_l ( \mathrm{d} \sF^l_1 \wedge \mathrm{d} \sF^l_2)(p) = \sum_{l=1}^k w_l \mathrm{Im}(\mathrm{d}\overline{\sF^l} \wedge \mathrm{d} \sF^l)(p),
	\]
	since $\sF(p)=0$ and $\nabla \sF = \mathrm{d} \sF$. The proof for the first part of the theorem is completed by noting that $|\dot{\rho}(\mathrm{i}) \sF|^2 = \sum_{l=1}^kw_l^2|\sF^l|^2$.
	In normal coordinates $x^1,\dots, x^{2m}$ around $p \in M_{\sF}$, we have $\nabla \mathrm{d} |\sF^l|^2(p) = 2 \mathrm{Re}(\partial_i \overline{\sF^l} \partial_j \sF^l)(p) dx^i \otimes dx^j$, and
	\[
	[\Delta, \gamma \llcorner] 
	%= 2\sum_{l=1}^k w_l(\partial_i \sF^l_1 \partial_j \sF^l_2 - \partial_i \sF^l_2 \partial_j \sF^l_1)  \iota_{dx^i} \iota_{dx^j} 
	= 2 \sum_{l=1}^k w_l \mathrm{Im}(\partial_i \overline{\sF}^l \partial_j \sF^l)  \iota_{dx^i} \iota_{dx^j}.
	\]
	It follows that
	$
	[\Delta, \gamma \llcorner]^m %= 2^m m! \pf(\partial_i \sF_1 \partial_j \sF_2 - \partial_i \sF_2 \partial_j \sF_1) 
	= 2^m m! \pf \mathrm{Im}(\partial_i \overline{\sF} \partial_j \sF),
	$
	which completes the proof for the second part of the theorem.
\end{proof}

\begin{rmk}
	If $M_{\sF}$ is discrete, then both $X$ and $\sF$ can be linearized in the normal coordinates $x^1, \cdots, x^{2m}$ around any point $p \in M_{\sF}$. We can write
	\[
	X(p) = \sum_{j=1}^{m}\lambda_j(p)(x^{2j}\partial_{2j-1}-x^{2j-1}\partial_{2j}), \quad \sF^l(p) = \sum_{j=1}^{2m}c^l_{j}(p)x^{j},
	\]
	where $c^l_j(p)$ are complex numbers. Since $\sF$ is $\U(1)$-equivariant, we must have $X(\sF)(p) = w_l \mathrm{i} \sF$, which implies that
	\begin{align*}
		\begin{cases}
			c_{2j-1}^l = \pm \mathrm{i} c^l_{2j}, & \text{if}~\lambda_j = \pm w_l; \\
			c_{2j-1}^l = c^l_{2j}=0, & \text{if}~\lambda_j \neq \pm w_l.
		\end{cases}
	\end{align*}
	Let $\Lambda_X\coloneqq\cup_{p \in M_X}\{\lambda_1(p), \dots, \lambda_m(p)\}$. Since $\Hess |\dot{\rho}(\mathrm{i})\sF|^2$ is non-degenerate at each $p \in M_X$, we conclude that
	\[
	\Lambda_X \subset \{w_1, \dots, w_k\}.
	\]
	In particular, we must have $k \geq |\Lambda_X|$.
\end{rmk}

\section{Applications and discussions}

To apply Theorems \ref{pvloc} and \ref{pvloccoh}, we need to find a way to extend a $\Delta$-closed multivector field to a equivariantly closed multivector field. Let us consider a function of the form $e^S \in C^{\infty}(M)$ and assume that we can find a bivector field $I_X \in \MV^{-2}(M)$ such that
\[
\Delta_{\mathfrak{u}(1)} e^{S_{eq}} = 0,
\quad
S_{eq}\coloneqq S + \phi I_X.
\]
Such $I_X$ exists if and only if 
\[
\Delta(S_{eq}) + \frac{1}{2}\{S_{eq}, S_{eq}\} - \phi X = 0.
\]
This equation can be broken into three independent equations:
\begin{align}
	&\Delta(S) + \frac{1}{2}\{S,S\} = 0, \label{eq1}\\
	&\Delta(I_X) + \{S, I_X\} = X, \label{eq2}\\
	&\{I_X, I_X\}=0. \label{eq3}
\end{align}

\subsection{Poisson geometry}

For the BV algebra $\MV(M)$, the quantum master equation \eqref{eq1} is satisfied since the restrictions of $\Delta$ and $\{\cdot, \cdot\}$ to $C^{\infty}(M)$ vanish. Equation \eqref{eq3} tells us that $I_X$ defines a Poisson structure $\pi$ on $M$.  Equation \eqref{eq2} is the most interesting one. In particular, it implies that
\[
\mathrm{Lie}_X(\pi) = \{X, \pi\} = \{\Delta(\pi), \pi\} + \{\{S,\pi\} , \pi\} = [\Delta, \{\pi, \pi\}] + \frac{1}{2} \{S, \{\pi,\pi\}\}= 0.
\]
Therefore, the $\mathrm{U}(1)$-action on $(M,\pi)$ is Poisson.
It follows from Theorem \ref{pvloc} that
\begin{thm}\label{ap1}
	Let $(M, \pi)$ be a $2m$-dimensional orientable Poisson compact manifold endowed with a Poisson $\mathrm{U(1)}$-action. Let $\vol$ be a $\mathrm{U}(1)$-invariant volume form on $M$. Supposing that
	\begin{align}\label{key}
		[X]=[X_{\vol}] \in H^1_{\pi}(M),
	\end{align}
	where $H^{\bullet}_{\pi}(M)$ are the Poisson cohomology groups of $M$, $X$ is the fundamental vector field of the $\mathrm{U}(1)$-action, and $X_{\vol}=\Delta(\pi)$ is the divergence of $\pi$ with respect to $\vol$, then we can find a function $h$ on $M$ satisfying
	\begin{align*}
		 X = X_{\vol} + X_h, 
	\end{align*}
	where $X_h = \{h, \pi\}$ is the Hamiltonian vector field of $h$. If the $\mathrm{U}(1)$-action has isolated fixed points, then
	\begin{equation}\label{locpoi}
		\int_M e^h \vol  = \frac{(-2\pi)^m}{m!} \sum_{p \in M_{X}}e^{h(p)}\frac{\langle \vol, \pi^{\wedge m}\rangle (p)}{\lambda_1(p)\cdots\lambda_m(p)}.
	\end{equation}
\end{thm}
\begin{rmk}
	If $\pi = \omega^{-1}$ and $\vol = \omega^{\wedge m}/m!$, where $\omega$ is a symplectic structure on $M$, \eqref{locpoi} recovers the Duistermaat-Heckman localization formula in symplectic geometry.
\end{rmk}

%Let us consider the following reducible solution to \eqref{eq2}:
%\[
%\Delta(I_X)=0, \{S,I_X\}=X.
%\]
%In the settings of $\MV(M)$, this is equivalent to saying that the Poisson structure $\pi$ is unimodular and the $\mathrm{U}(1)$-action is Hamiltonian. We obtain the following interesting result as a corollary of Theorem \ref{ap1}.
In particular, if the Poisson structure $\pi$ is unimodular, i.e., $[X_{\vol}]=0$, then \eqref{key} implies that the $\mathrm{U}(1)$-action on $(M,\pi)$ is Hamiltonian.

\begin{cor}\label{unipoissrank}
	Let $(M,\pi)$ be a $2m$-dimensional compact unimodular Poisson manifold. If $(M,\pi)$ has a Hamiltonian $\mathrm{U}(1)$-action with a set $M_{X}$ of isolated fixed points, then one can find $p \in M_{X}$ such that $\mathrm{rank}(\pi_p)=2m$. 
\end{cor}
\begin{proof}
	By assumption, one can find a $\mathrm{U}(1)$-invariant volume form $\vol$ on $M$ such that $X_{\vol}=0$. Theorem \ref{ap1} then implies that 
	\[
	\sum_{p \in M_{X}}\exp(h(p))\frac{\langle \vol, \pi^{\wedge m}\rangle (p)}{\lambda_1(p)\cdots\lambda_m(p)}  \neq 0,
	\]
	where $h$ is any Hamiltonian function of the $\mathrm{U}(1)$-action. This is possible only if $\pi$ has rank $2m$ at one of the fixed points of the $\mathrm{U}(1)$-action. 
\end{proof}
\begin{cor}
	Let $(M,\pi)$ be as in Corollary \ref{unipoissrank}. If $\pi$ is regular, then it must be sympelctic.
\end{cor}

\subsection{Equivariant AKSZ theory}

The solutions
\[
\Delta(I_X) = 0, \quad \{S,I_X\}=X
\]
to \eqref{eq2} also play an interesting role in equivariant AKSZ field theories \cite{bonechi2020equivariant}. In fact, the functional
\[
I_X = \int_{\Sigma} \mathbf{P} \iota_{X_{\Sigma}} \mathbf{Q}. 
\]
constructed in \cite{bonechi2023towards} is an example of such solution. Here, $\mathbf{P}$ and $\mathbf{Q}$ are the AKSZ superfields, $S$ is the AKSZ action functional, $X_{\Sigma}$ is a fundamental vector field corresponding to a $\U(1)$-action on the source manifold $\Sigma$, $\iota_{X_{\Sigma}}$ is the contraction by $X_{\Sigma}$, and $X$ is the vector field over the mapping space induced by $X_{\Sigma}$. 

%We will discuss the applications of Theorem \ref{pvloc} in AKSZ field theories in more detail in future work.

\subsection{Cohomological field theory}

Let us briefly outline the connections between Theorem \ref{pvloccoh} and the zero-dimensional cohomological field theory determined by the $G$-equivariant map $\sF\colon M \rightarrow W$. The CohFT configuration space is the following differential graded manifold: 
\[
\mathcal{E}:=(T[1](M \times W[-1] \times \mathfrak{g}[1]), Q).
\]
%where $Q$ is a combination of the Kalkman differential and the Koszul differential. More precisely, 
Note that the algebra of superfunctions on $\mathcal{E}$ can be identified as
\[
C^{\infty}(\mathcal{E}) \cong W(\mathfrak{g}) \otimes \Omega(M) \otimes \Omega(W),
\]
where $\Omega(W) \coloneqq \Lambda (W^{\vee}) \otimes \Sym (W)$. The cohomological vector field $Q$ is the Kalkman differential \cite{kalkman1993brst}. % of the BRST model of the equivariant cohomology of $M \times W$. 
Let $(x^{\mu}, \psi^{\mu}, \chi^{i}, b_i, \theta^a, \phi^a)$ be local coordinate functions on $\mathcal{E}$. The CohFT action functional is defined as
\[
S = Q(i\langle \chi, b \rangle + \langle \chi, \sF \rangle),
\]
where $\langle \cdot, \cdot \rangle$ is the canonical pairing between $W^*$ and $W$. Consider the path integral 
\[
Z = \int_{\mathcal{E}} \mu e^{\mathrm{i}tS} P,
\]
where $\mu$ is the canonical Berezinian on $\mathcal{E}$, $P$ is a $Q$-closed superfunction on $\mathcal{E}$. $Z$ and $Q$ can be transformed into \eqref{pi} and the equivariant BV Laplacian $\Delta_{\mathfrak{g}}$, respectively, via an equivariant extension of the odd Fourier transform \cite{qiu2011introduction}
\[
C^{\infty}(T[1](M\times W[-1])) \rightarrow C^{\infty}(T^*[-1](M\times W[-1])),
\]
followed by integrating out the $\chi$ and $b$ variables.\footnote{This is well-defined because $e^{\mathrm{i}tS}$ is Gaussian with respect to $b$.}

Cohomological field theories are both mathematically and physically more compelling in non-zero dimensions. While the Berezinian on the infinite-dimensional CohFT configuration space is not well-defined, the moduli space $\sF^{-1}(0)/G$ is finite-dimensional for a nice $\sF$, and the CohFT path integral can be perturbatively well-defined. Furthermore, the infinite-dimensional (equivariant) BV Laplacian can be rigorously defined through an appropriate regularization procedure, such as heat kernel regularization \cite{costello2022renormalization}. This is a key advantage of the (equivariant) BV theory over (equivariant) de Rham theory in applications to (perturbative) quantum field theory. Therefore, it would be interesting to extend the BV equivariant localization principle introduced in this work to infinite-dimensional settings and apply it to various CohFTs, such as Donaldson--Witten and Seiberg--Witten theories, where numerous elegant localization formulas have been derived by mathematicians and physicists \cite {blau1995localization, constantinescu1998cicular,cordes1995lectures, nekrasov2003seiberg, pestun2012localization, pestun2017localization, vajiac2000derivation}.

\vspace{10pt}

This work opens several promising avenues for future research. Notably, higher and non-abelian generalizations of the framework merit investigation, as they are expected to reveal the deeper capabilities of the Batalin-Vilkovisky formalism. For example, in the present work, we implicitly rely on the fact that $M$ serves as a gauge-fixing Lagrangian submanifold of $T^*[-1]M$. However, for a general odd symplectic manifold, there is greater flexibility in choosing the gauge-fixing condition. We plan to explore these directions in future studies.

\end{document}